\documentclass{article}

\PassOptionsToPackage{numbers, compress}{natbib}

\usepackage[preprint]{neurips_2026}

\usepackage[utf8]{inputenc} 
\usepackage[T1]{fontenc}    
\usepackage{hyperref}       
\usepackage{url}            
\usepackage{booktabs}       
\usepackage{amsfonts}       
\usepackage{nicefrac}       
\usepackage{microtype}      
\usepackage{xcolor}         
\usepackage{amsmath}
\usepackage{amssymb}
\usepackage{multirow}
\usepackage{graphicx}
\usepackage{wrapfig}
\usepackage{algorithm}
\usepackage{algpseudocode}
\usepackage{tikz}
\usepackage{amsthm}
\newtheorem{proposition}{Proposition}
\usetikzlibrary{positioning,arrows.meta,calc,fit,backgrounds,shapes.geometric,patterns}

\title{AffectCodec: Emotion-Preserving Neural Speech Codec with Block-Diagonal Residual FSQ}

\author{%
  Zhaoyang Meng \quad Zhengyao Ma \quad Kecan Mao \quad Yingming Gao \quad Ya Li\thanks{Corresponding author}\\
  Beijing University of Posts and Telecommunications\\
  \texttt{\{mengzy, mazhyao, mao\_kecan, yingming.gao, yli01\}@bupt.edu.cn}
}

\begin{document}

\maketitle

\begin{abstract}
Neural speech codecs have become the discrete interface between raw audio and speech language models, yet they remain optimized primarily for acoustic reconstruction fidelity, which leaves emotion-relevant cues vulnerable to being discarded during quantization, limiting the affective capacity of downstream models. We trace this degradation to two mechanisms: reconstruction-driven bit allocation under limited bitrate and cross-stream leakage in concatenation-based codecs, where acoustic gradients can overwrite nominally emotion-reserved dimensions. We propose AffectCodec, an emotion-preserving neural speech codec built on Block-Diagonal Residual Finite Scalar Quantization (BD-RFSQ). By imposing block-diagonal input and output projections over emotion and acoustic subspaces, BD-RFSQ transforms bit allocation from implicit and loss-driven to explicit and structurally guaranteed, while still preserving a flat token interface for downstream speech language models. AffectCodec further combines this structurally constrained quantizer with multi-granularity emotion conditioning and multi-rate training, enabling robust affect preservation at low bitrates. Experiments across multiple emotional speech benchmarks show that AffectCodec substantially improves emotion preservation, especially in the low-bitrate regime, while maintaining competitive acoustic quality and intelligibility. These results suggest that structurally protected quantization is an effective principle for preserving emotion-relevant information and may provide a general route toward attribute-aware neural speech compression.
\end{abstract}

\section{Introduction}\label{sec:intro}

Speech Language Models (SLMs), such as VALL-E~\cite{valle}, CosyVoice~\cite{cosyvoice2}, and Moshi~\cite{moshi}, have repositioned neural speech codecs from standalone compression modules to discrete tokenizers for spoken language modeling. This shift makes codec representations a critical bottleneck: downstream models can only exploit the information preserved in discrete codec tokens. As SLMs move toward emotionally sensitive applications, including empathetic dialogue, mental health screening, and expressive dubbing, preserving emotion-relevant cues during tokenization becomes essential. Once affective information is discarded by the codec, it cannot be reliably recovered from the resulting discrete representation. 

However, existing speech codecs are primarily designed for perceptual reconstruction rather than emotion preservation. As shown in Fig.~\ref{fig:intro}, passing speech through an encode--quantize--decode pipeline degrades SER Macro-F1 on IEMOCAP from about 64\% to around 53\% at low bitrates, with a persistent gap even at 6.0\,kbps. EMO-Codec~\cite{emocodec} corroborates this finding across 10 codecs and 6 datasets, yet neither the mechanism behind this degradation nor a principled remedy has been established.

We argue that this emotion loss is not merely a side-effect of lossy compression but stems from the lack of structural protection for affective information inside the quantizer. Indeed, our preliminary analysis across IEMOCAP, CREMA-D, and ESD reveals a consistent mismatch between 
\begin{wrapfigure}{r}{0.45\textwidth}
    \centering
    \includegraphics[width=0.43\textwidth]{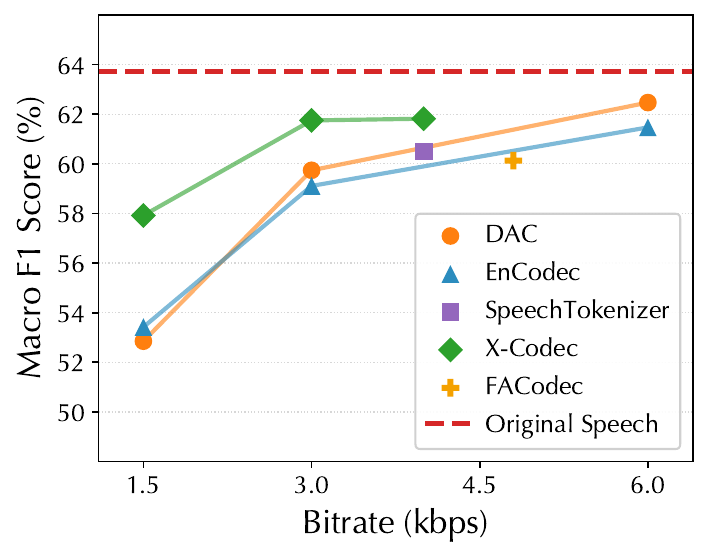}
    \caption{SER Macro-F1 of neural codecs across bitrates on IEMOCAP. The red dashed line denotes performance on original speech.}
    \label{fig:intro}
    \vspace{-0.5em}
\end{wrapfigure}
standard acoustic-quality metrics (STOI, ViSQOL) and emotion retention, indicating that emotion-relevant cues are not reliably preserved as a byproduct of acoustic reconstruction quality. We identify two key causes. \textbf{(1)~Reconstruction-driven bit allocation.} Standard codec objectives (mel-spectrogram, STFT, adversarial losses) prioritize broadband acoustic fidelity and are only weakly aligned with emotion-relevant cues such as pitch trajectory and energy dynamics. 
At high bitrates some affective information survives incidentally, but under capacity pressure it is easily sacrificed. \textbf{(2)~Cross-stream leakage.} A natural fix is to concatenate a pretrained emotion representation with the acoustic latent before quantization. However, when the quantizer uses fully connected projections, each quantization dimension mixes both streams, and dominant reconstruction gradients repurpose nominally emotion-reserved dimensions for acoustic fidelity (verified empirically in Appendix~\ref{app:leakage}).

These observations call for a quantizer whose bit allocation is structurally guaranteed rather than loss-driven, and whose emotion reservation is enforced architecturally rather than through loss balancing alone. We propose \textbf{AffectCodec}, a neural speech codec built around Block-Diagonal Residual Finite Scalar Quantization (BD-RFSQ). BD-RFSQ wraps each residual FSQ stage with learnable input/output projections constrained to be block-diagonal over emotion and acoustic subspaces: emotion dimensions are projected only into the reserved emotion partition, preventing cross-stream overwriting by construction. At the same time, each stage emits a single composite token, preserving a flat token interface compatible with downstream SLMs. We make the following contributions:

\begin{itemize}
    \item We propose \textbf{BD-RFSQ}, which turns bit allocation from an implicit, loss-driven process into an explicit and structurally guaranteed design, while preserving unified per-stage tokens compatible with flat-token speech language models.

    \item We design a \textbf{multi-rate training strategy} comprising a multi-rate reconstruction task that supervises intermediate residual-stage outputs with both mel-reconstruction and emotion cycle-consistency losses, together with a biased stage dropout that concentrates training on low-bitrate operating points where emotion degradation is most severe.

    \item We present \textbf{AffectCodec}, which integrates BD-RFSQ with the multi-rate training strategy and multi-granularity emotion conditioning. Experiments across multiple emotional speech benchmarks show that AffectCodec substantially improves emotion preservation---especially at low bitrates---while maintaining competitive acoustic quality and intelligibility.

\end{itemize}

\section{Related Work}\label{sec:related}

\paragraph{Neural speech codecs.}
Modern neural speech codecs largely build on the VQ-VAE
framework~\cite{vqvae} and the high-fidelity codec paradigm established by
SoundStream~\cite{soundstream}, which combines a fully convolutional
encoder--decoder, residual vector quantization (RVQ), and adversarial training.
EnCodec~\cite{encodec} further improves this recipe with multi-scale STFT
discriminators and a gradient-based loss balancer for stabilizing heterogeneous
reconstruction objectives. DAC~\cite{dac} advances reconstruction quality
through snake activations, factorized $\ell_2$-normalized projections for
mitigating codebook collapse, and quantizer dropout for supporting multiple
bitrates within a single model. HiFi-Codec~\cite{hificodec} introduces
Group-RVQ, which partitions encoder features into parallel groups and quantizes
each group with a smaller RVQ stack, reducing codebook complexity while
maintaining competitive fidelity. Despite these advances, most neural speech codecs are still optimized primarily
for perceptual reconstruction, using objectives such as mel-spectrogram or
multi-scale STFT losses and evaluating quality with metrics such as ViSQOL. Emotion preservation, however, has rarely been
treated as a primary codec-design objective. The recent EMO-Codec
benchmark~\cite{emocodec} exposes this limitation: across multiple codecs and
datasets, codec compression substantially degrades downstream speech emotion
recognition, especially at low bitrates. These findings suggest that emotion
loss should be addressed at the codec level, rather than being left entirely to
downstream emotion compensation.

\paragraph{Discrete quantization for neural audio codecs.}
A complementary line of work studies the quantization mechanism itself. Vanilla
VQ-VAE~\cite{vqvae} represents each frame with a learned codeword, but it is
prone to codebook collapse and requires exponentially large codebooks to scale
to high bitrates. RVQ~\cite{soundstream} addresses the latter issue by stacking
multiple small VQ layers, each quantizing the residual left by the previous
layer, thereby achieving a large effective code space with tractable codebooks.
Codebook utilization can be further improved by EMA updates, commitment losses,
and factorized projections into a low-dimensional codebook space~\cite{dac}.

Finite Scalar Quantization (FSQ)~\cite{fsq} replaces learned codebooks with
dimension-wise rounding to fixed scalar levels, yielding simple and stable
tokenization with high code utilization. This property has made FSQ attractive
for speech language modeling~\cite{cosyvoice2}, while related lookup-free
quantizers further simplify discrete representations~\cite{lfq}. To increase
capacity, Residual FSQ (RFSQ) stacks multiple FSQ stages over residuals, but
later stages often suffer from residual magnitude decay and underuse the scalar
grid. Robust RFSQ~\cite{rrfsq} alleviates this issue through stage-wise
normalization. However, existing RFSQ formulations still operate directly in
the low-dimensional FSQ grid space, making them poorly suited for
high-dimensional speech latents. Compressing an encoder representation into
such a narrow quantization space can severely limit expressiveness, motivating
factorized projections between high-dimensional latent spaces and compact FSQ
spaces.

\paragraph{Semantic- and attribute-aware codecs.}
Recent codecs have begun to incorporate representations beyond purely acoustic
reconstruction. X-Codec~\cite{xcodec} concatenates self-supervised semantic
features with acoustic encoder outputs before quantization to improve phonetic
intelligibility for LLM-based speech generation.
SpeechTokenizer~\cite{speechtokenizer} uses cross-layer distillation to align
the first RVQ codebook with HuBERT~\cite{hubert} features, encouraging later
codebooks to capture timbre and acoustic details. FACodec~\cite{naturalspeech} 
assigns a dedicated quantizer to each of content,
prosody, timbre, and acoustic-detail streams. Its prosody stream captures
intonation-related variation that overlaps with some affective cues, but prosody
alone is not equivalent to emotion. Moreover, the use of separate quantizers per
attribute produces structurally heterogeneous tokens that require specialized
downstream models (e.g., factorized diffusion), limiting compatibility with
flat-token speech language model architectures.

\section{AffectCodec}\label{sec:method}

\subsection{Emotion-Acoustic Dual-Path Architecture}\label{sec:arch}

Affective cues in speech are carried by pitch, energy, speaking rate, and long-range prosodic contours. Compared with the broadband spectral details that standard codec objectives optimize for, these cues are typically low-dimensional and temporally smooth. Reconstruction losses such as mel and STFT distances therefore provide only indirect supervision for emotion preservation: a model can achieve low spectral distortion while still altering the pitch dynamics or prosodic trajectories that are critical for perceived emotion. This limitation is further amplified by the multi-scale nature of affective cues in speech: utterance-level emotional tone shapes global prosody, whereas frame-level pitch and energy variations convey fine-grained affective nuances. Standard codec encoders, optimized primarily for broadband acoustic fidelity, are therefore not explicitly incentivized to capture emotion information at either granularity. This motivates a dual-path encoder that models affective information through a dedicated pathway and fuses multi-scale emotion cues into the codec latent, rather than relying on reconstruction losses alone. Fig.~\ref{fig:overview} (right) illustrates the overall architecture. 

\paragraph{Acoustic and emotion encoders.}
The acoustic branch follows a DAC-style convolutional encoder $\mathcal{E}_{ac}$
with strides $[2,4,5,8]$ and a total downsampling factor of 320, mapping the
waveform to a high-dimensional acoustic representation
$\mathbf{A} = \mathcal{E}_{ac}(\mathbf{x}) \in \mathbb{R}^{d_a' \times T}$,
where $T = L/320$. In parallel, a frozen emotion2vec encoder extracts
frame-level affective features, which are aligned to the codec frame rate by a
lightweight convolutional adapter:
\begin{equation}
    \mathbf{E}
    =
    \mathcal{E}_{em}(\operatorname{emo2vec}(\mathbf{x}))
    \in \mathbb{R}^{d_e' \times T}.
    \label{eq:emotion_enc}
\end{equation}
Using a frozen pretrained emotion teacher provides a high-quality affective
signal and prevents the emotion representation from drifting under the pressure
of reconstruction objectives during training.

\paragraph{Multi-granularity emotion conditioning.}
AffectCodec models affective information at two complementary scales.
At the \textbf{coarse level}, the Coarse-granularity Emotion Modulation (CEM) module
extracts a global emotion embedding $\mathbf{e}_g = \mathrm{AttnPool}(\mathbf{E})$
via attentive pooling and fuses it into the acoustic pathway through FiLM~\cite{film}
modulation:
\begin{equation}
\mathbf{A}_f
= \gamma \odot \mathbf{A} + \beta,
\quad
\gamma = g(\mathbf{e}_g),
\quad
\beta = h(\mathbf{e}_g),
\end{equation}
where $g(\cdot)$ and $h(\cdot)$ are two-layer linear projections.
At the \textbf{fine level}, the frame-level emotion features $\mathbf{E}$ are channeled into a
dedicated quantization pathway in BD-RFSQ (Sec.~\ref{sec:bdrfsq}),
preserving the local pitch, energy, and prosodic variations that carry
fine-grained emotional expressiveness.

\begin{figure}[t]
    \centering
    \includegraphics[width=0.95\textwidth]{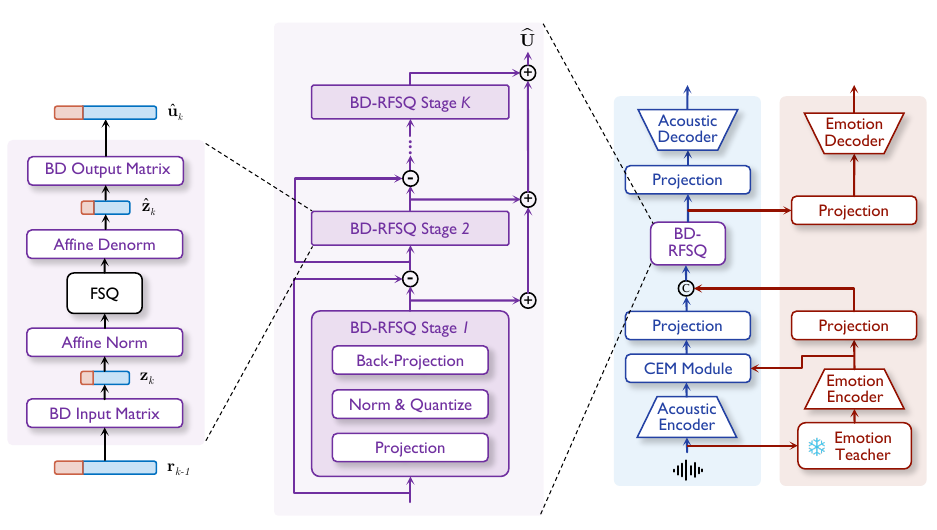}
    \caption{%
        \textbf{Left}: Internal structure of a single BD-RFSQ stage.
        The residual $\mathbf{r}_{k-1}$ is projected to the compact FSQ space
        via a block-diagonal input matrix (red: emotion partition; blue:
        acoustic partition), affine-normalized, scalar-quantized by FSQ,
        affine-de-normalized, and mapped back to the latent space via a
        block-diagonal output matrix, yielding the stage reconstruction
        $\widehat{\mathbf{u}}_k$.
        \textbf{Center}: BD-RFSQ chains $K$ such stages with residual
        connections; each stage quantizes the current residual and subtracts
        its reconstruction before passing the remainder to the next stage.
        \textbf{Right}: AffectCodec architecture. A frozen emotion teacher and
        a DAC-style acoustic encoder form a dual-path front-end; the CEM module
        fuses coarse-grained emotion cues into the acoustic pathway, and the
        concatenated representation is discretized by BD-RFSQ. Separate
        acoustic and emotion decoders reconstruct the waveform and supervise
        the emotion partition, respectively.%
    }
    \label{fig:overview}
\end{figure}

\subsection{Block-Diagonal Residual FSQ}\label{sec:bdrfsq}

In this subsection, we first introduce factorized RFSQ, which adapts the factorized codebook projection technique~\cite{vqgan} to naive RFSQ, improving its training stability and reconstruction quality. Building upon this foundation, we then present BD-RFSQ, the core quantizer design of AffectCodec.

\paragraph{Factorized RFSQ.}
FSQ quantizes each dimension independently into a small number of scalar
levels, so the quantization grid dimension is typically very small. Unlike VQ,
which can employ high-dimensional codebooks, forcing a high-dimensional encoder
latent through a narrow FSQ grid would severely limit representational
capacity. To address this issue, we propose factorized RFSQ by introducing the
factorized codebook technique into residual FSQ. The factorized codebook
technique~\cite{vqgan,dac} was originally proposed for VQ-based quantizers, which uses a 
low-dimensional lookup space to decouple code selection from high-dimensional 
code embedding, improving codebook utilization and reconstruction quality.
Factorized RFSQ adapts this principle by wrapping each FSQ stage with learnable input and output projections
$\pi_{\mathrm{in}}: \mathbb{R}^d \to \mathbb{R}^f$ and
$\pi_{\mathrm{out}}: \mathbb{R}^f \to \mathbb{R}^d$, where $d \gg f$.
Residuals are thus maintained in the original high-dimensional latent space
where the encoder representation is expressive, while scalar quantization
operates in a compact space where FSQ is effective. This factorization
substantially improves reconstruction quality over naive RFSQ.

\paragraph{Block-diagonal constraint.}
Factorized RFSQ alone does not structurally protect any designated
attribute. With fully connected projections, each quantization dimension reads
from and writes to all input channels, so an intended emotion/acoustic
partition can be overwritten by dominant reconstruction gradients. To enforce partition integrity, we
further constrain the input and output projections of factorized RFSQ to be
block-diagonal with respect to the emotion/acoustic split, yielding
Block-Diagonal Residual Finite Scalar Quantization (BD-RFSQ).
The emotion and acoustic features are first projected into a partitioned latent space:
\begin{equation}
\mathbf{U}_e = \phi_e(\mathbf{E}) \in \mathbb{R}^{d_e \times T}, \quad
\mathbf{U}_a = \phi_a(\mathbf{A}_f) \in \mathbb{R}^{d_a \times T}, \quad
\mathbf{U} = \operatorname{Concat}(\mathbf{U}_e, \mathbf{U}_a).
\end{equation}
Let $\mathbf{r}_0=\mathbf{U}$ denote the initial residual,
$d=d_e+d_a$ the total latent dimension, and
$f=f_e+f_a$ the compact FSQ dimension. Each BD-RFSQ stage $k$
proceeds in three steps: {project}, {normalize-and-quantize}, and
{back-project}. \textbf{(i) Block-diagonal projection.}
The residual is projected via:

\begin{equation}
    \mathbf{z}_k \in \mathbb{R}^{f}
    = \pi_{\mathrm{in}}^{(k)}(\mathbf{r}_{k-1}),
    \qquad
    \pi_{\mathrm{in}}^{(k)}
    =
    \begin{bmatrix}
    \pi_{\mathrm{in},e}^{(k)} & 0 \\
    0 & \pi_{\mathrm{in},a}^{(k)}
    \end{bmatrix},
\end{equation}
where the block-diagonal form ensures that emotion and acoustic channels are
read independently, enforcing stream-level separation within the quantizer.
\textbf{(ii) Affine normalization and quantization.}
To address residual magnitude decay in later stages~\cite{rrfsq}, we apply a
learnable per-dimension affine transformation on the residual before quantization:
\begin{equation}
    \widetilde{\mathbf{z}}_k
    = \mathbf{s}_k \odot (\mathbf{z}_k - \mathbf{b}_k)
\end{equation}

where $\mathbf{s}_k, \mathbf{b}_k \in \mathbb{R}^{f}$ are per-dimension scale factor and bias
that normalize the residual to cover the effective range of FSQ grid, $\odot$ denotes element-wise multiplication.
The scale is parameterized as
$\mathbf{s}_k = \mathrm{softplus}(\boldsymbol{\ell}_k) + \epsilon$, where
$\boldsymbol{\ell}_k \in \mathbb{R}^{f}$ is a learnable parameter and
$\epsilon$ is a small positive floor; the $\mathrm{softplus}$ ensures
$\mathbf{s}_k > 0$ with smooth, non-vanishing gradients everywhere.
Unlike the data-dependent LayerNorm
in~\cite{rrfsq}, our affine uses only fixed model parameters, naturally
supporting end-to-end training and index-only decoding. The normalized vector
$\widetilde{\mathbf{z}}_k \in \mathbb{R}^{f}$ is then quantized:
$\widehat{\widetilde{\mathbf{z}}}_k, I_k
= \mathrm{FSQ}(\widetilde{\mathbf{z}}_k)$,
where $\widehat{\widetilde{\mathbf{z}}}_k \in \mathbb{R}^{f}$ is the
quantized vector and $I_k$ is the discrete code index emitted by stage $k$.
\textbf{(iii) Back-projection and residual update.}
The quantized vector is inverse-normalized and mapped back via a block-diagonal
output projection:
\begin{equation}
    \widehat{\mathbf{u}}_k
    =
    \pi_{\mathrm{out}}^{(k)}
    \!\left(
    \widehat{\widetilde{\mathbf{z}}}_k \oslash \mathbf{s}_k
    + \mathbf{b}_k
    \right),
    \quad
    \pi_{\mathrm{out}}^{(k)}
    =
    \begin{bmatrix}
    \pi_{\mathrm{out},e}^{(k)} & 0 \\
    0 & \pi_{\mathrm{out},a}^{(k)}
    \end{bmatrix},
\end{equation}
where $\oslash$ denotes element-wise division. The residual is updated as
$\mathbf{r}_k = \mathbf{r}_{k-1} - \widehat{\mathbf{u}}_k$, and the final
quantized latent is $\widehat{\mathbf{U}} = \sum_{k=1}^{K}\widehat{\mathbf{u}}_k$.

\paragraph{Structural guarantee and token format.}

With block-diagonal input/output projections and dimension-wise FSQ, emotion and
acoustic residuals are updated only within their respective partitions, which
provides a structural guarantee that emotion indices cannot be overwritten by
acoustic channels inside the quantizer (proved formally in Appendix~\ref{app:proofs}). Meanwhile, BD-RFSQ preserves a flat
token interface: each stage emits a
single composite index, where emotion and acoustic sub-indices occupy different
dimensions of the same token. Downstream speech language models therefore
receive a uniform sequence of per-stage tokens without needing heterogeneous
token handling. This contrasts with SpeechTokenizer~\cite{speechtokenizer} and
FACodec~\cite{naturalspeech}, which rely on stage- or attribute-specific token
structures and thus require specialized downstream modeling or generation
pipelines.

\subsection{Training Strategy}\label{sec:strategy}

\paragraph{Multi-rate reconstruction task.}
BD-RFSQ reserves emotion-specific capacity at every residual stage, but without
explicit supervision at intermediate stage counts, the model has little
incentive to maintain emotion fidelity or acoustic quality at low bitrates.
We address this with a multi-rate reconstruction loss. Let
$\widehat{\mathbf{U}}_m = \sum_{k=1}^{m}\widehat{\mathbf{u}}_k$ denote the
cumulative quantized latent after $m$ stages and
$\widehat{\mathbf{x}}_m=\mathcal{D}_{\mathrm{ac}}(\widehat{\mathbf{U}}_m)$.
For a set of target stage counts $\mathcal{S}_{\mathrm{mr}}$ corresponding to
the operating bitrates used at inference, we define
\begin{equation}
    \mathcal{L}_{\mathrm{mr}}
    =
    \sum_{m\in\mathcal{S}_{\mathrm{mr}}}
    w_m
    \left[
    \mathcal{L}_{\mathrm{mel}}(\widehat{\mathbf{x}}_m,\mathbf{x})
    +
    \eta\,
    \mathcal{L}_{\mathrm{cycle}}(\widehat{\mathbf{x}}_m,\mathbf{x})
    \right].
\end{equation}
Each intermediate output is decoded through the shared decoder and supervised
with both mel-reconstruction and emotion cycle-consistency losses, explicitly
encouraging the model to optimize emotion fidelity and acoustic quality at
low-bitrate operating points.

\paragraph{Biased stage dropout.}
Quantizer dropout~\cite{soundstream} enables variable-bitrate inference
by randomly truncating the number of active stages during training. DAC~\cite{dac}
improves upon the original uniform sampling by applying dropout with a fixed
probability, better balancing low-bitrate and full-bitrate quality. We further
extend this idea with a biased dropout distribution that concentrates
optimization effort on the low-bitrate regimes where emotion loss is most
severe: for each training sample, with a certain probability the active number
of stages is drawn from a categorical distribution biased toward fewer stages, while the
remaining samples use all $K$ stages. The dropout targets are aligned with the
multi-rate supervision points in $\mathcal{L}_{mr}$, so each operating bitrate
receives both dedicated training coverage and direct emotion-preservation
signals (details in Appendix~\ref{app:mr_config}).

\paragraph{Overall objective.}
The full training objective combines reconstruction, quantization, and
emotion-preservation terms. The reconstruction loss
$\mathcal{L}_{{rec}}$ aggregates multi-scale mel-spectrogram,
time-domain $L_1$, multi-scale STFT adversarial, and feature matching losses.
The emotion feature loss
$\mathcal{L}_{{emo}} =
\|\widehat{\mathbf{E}} - \mathbf{E}\|_2^2$ supervises the reserved emotion
partition in the emotion2vec feature space, where
$\widehat{\mathbf{E}} =
\mathcal{D}_{{em}}(\widehat{\mathbf{U}}_{1:d_e})$. To further penalize
affective distortion introduced by the decoder, we impose an emotion
cycle-consistency loss $\mathcal{L}_{{cycle}}$, defined as the cosine distance
between the temporally mean-pooled emotion2vec embeddings of the synthesized
waveform $\widehat{\mathbf{x}}$ and the original waveform $\mathbf{x}$.

All modules are trained jointly with the combined loss:
\begin{equation}
    \mathcal{L}
    =
    \mathcal{L}_{{rec}}
    + \alpha\, \mathcal{L}_{{cm}}
    + \beta\, \mathcal{L}_{{emo}}
    + \lambda\, \mathcal{L}_{{cycle}}
    + \delta\, \mathcal{L}_{{mr}},
    \label{eq:total_loss}
\end{equation}
where $\alpha, \beta, \lambda, \delta$ are scalar weighting coefficients.
Unlike vanilla FSQ, we include a commitment loss in BD-RFSQ to keep the pre-quantization values close to their quantized counterparts, stabilizing straight-through gradient estimation.

\section{Experiments}\label{sec:experiments}

\subsection{Experimental Setup}\label{sec:setup}

\paragraph{Datasets.}
We train AffectCodec on a mixture of general and emotional speech. The general
corpus is LibriSpeech~\cite{librispeech} (960\,h of read English) and the emotional corpus is the IEMOCAP~\cite{iemocap} training
split (approximately 10\,h of scripted and improvised dialogue across four emotion classes: angry, happy, neutral, sad).
All audio is resampled to 16\,kHz. We evaluate on three datasets that span different speaker populations, recording conditions, and emotion taxonomies: IEMOCAP (10 actors, 4 classes, scripted and improvised dialogue), CREMA-D~\cite{crema} (91 actors with diverse demographics, 6 classes), and ESD~\cite{esd} (10 English-language speakers selected from 20 total, 5 classes, studio-recorded). Together they cover spontaneous and read speech, small- and large-scale speaker pools, and partially overlapping emotion labels.

\paragraph{Evaluation metrics.}
Our primary emotion metric is the Emotion Degradation Rate (EDR),
defined as the relative drop in SER F1 after codec reconstruction, i.e.,
$\mathrm{EDR} = (F_1(\mathbf{x}) - F_1(\widehat{\mathbf{x}})) / F_1(\mathbf{x}) \times 100\%$.
We report both Macro EDR (treating all classes equally) and Weighted EDR
(reflecting the natural class distribution). To avoid circular evaluation
with the emotion2vec teacher, we average EDR over three independently trained
SER classifiers built on frozen HuBERT-Large~\cite{hubert},
WavLM-Large~\cite{wavlm}, and Wav2Vec\,2.0-Large~\cite{wav2vec2} features,
each following the S3PRL/SUPERB benchmark design~\cite{superb}. We
additionally report V/A/D MSE---the mean squared error of predicted Valence,
Arousal, and Dominance between original and reconstructed speech, using
Wav2Vec2-Large fine-tuned on MSP-Podcast~\cite{msp}---to capture continuous
affective distortion that categorical EDR may miss. For acoustic quality we
use ViSQOL~\cite{visqol} and STOI~\cite{stoi}; for intelligibility we
report Word Error Rate (WER) obtained from Whisper-Large-v3~\cite{whisper}.

\paragraph{Baselines.}
We compare against four publicly available codecs covering the major design families: EnCodec~\cite{encodec} and DAC~\cite{dac} (standard RVQ), SpeechTokenizer~\cite{speechtokenizer} (semantic-distilled RVQ), and X-Codec~\cite{xcodec} (semantic-concatenation RVQ). All baselines use official pretrained checkpoints and are evaluated by truncating RVQ stages to match target bitrates. EnCodec and DAC support up to 6.0\,kbps; SpeechTokenizer and X-Codec both support up to 4.0\,kbps. None includes an explicit emotion preservation mechanism.

\paragraph{Implementation details.}
AffectCodec builds on the DAC encoder-decoder backbone with Snake activations~\cite{bigvgan} and encoder strides $[2,4,5,8]$, yielding a frame rate of 50\,Hz at 16\,kHz input. BD-RFSQ uses $K{=}8$ residual stages with a latent partition $(d_e, d_a) = (256, 768)$ projected to FSQ dimensions $(f_e, f_a) = (3, 6)$ (selected via rate-distortion search; Appendix~\ref{app:rd_search}) and scalar levels $\mathbf{L} = [2,2,2,4,4,4,4,4,4]$, where the first three dimensions form the emotion partition and the remaining six the acoustic partition. This yields $2^3 \cdot 4^6 = 2^{15}$ codes per stage and $15 \times 8 \times 50 = 6.0$\,kbps at full depth. The emotion2vec encoder is frozen throughout training. Additional hyperparameters and a pseudocode summary of BD-RFSQ are provided in Appendices~\ref{app:details} and~\ref{app:algorithm}. For the ablation study, each variant is trained
with the same schedule and hyperparameters as the full model. Only the
component under test is modified; all other settings remain identical.

\subsection{Comparative Evaluation}\label{sec:comparison}

\paragraph{Emotion preservation.}
Table~\ref{tab:main} reports EDR and V/A/D MSE across three datasets and
three bitrates. AffectCodec achieves the lowest Macro EDR in 7 of 9
dataset--bitrate conditions and ranks second in the remaining two
(3.0\,kbps on CREMA-D and ESD), where X-Codec and SpeechTokenizer lead by
small margins (8.51\% vs.\ 9.25\% and 8.75\% vs.\ 9.66\%, respectively).
The advantage is most pronounced under tight capacity: at 1.5\,kbps on
IEMOCAP, AffectCodec attains 5.27\% Macro EDR versus 9.09\% for X-Codec
and 17.05\% for DAC; on CREMA-D the gap widens to 12.67\% vs.\ 26.72\%.
At 6.0\,kbps, where standard codecs already preserve emotion incidentally,
AffectCodec still leads on all three datasets (e.g., 0.85\% vs.\ DAC's
1.95\% on IEMOCAP), though margins narrow as expected. V/A/D MSE
largely corroborates the categorical findings: AffectCodec obtains the
lowest MSE in 8 of 9 conditions, with the sole exception again at
3.0\,kbps on CREMA-D (2.64 vs.\ X-Codec's 2.57). The strong results on
CREMA-D---91 speakers with diverse demographics, none seen during
training---indicate that BD-RFSQ's structural emotion reservation
generalizes to unseen speakers and recording conditions.

\begin{table}[t]
\centering
\caption{Emotion preservation comparison across bitrates and datasets. MEDR
and WEDR denote Macro and Weighted Emotion Degradation Rate (\%, lower is
better). MSE reports V/A/D mean squared error ($\times 10^{-3}$, lower is
better). Best results are \textbf{bolded}, and second-best results are \underline{underlined}. "--" indicates the model does not support that bitrate.}
\label{tab:main}
\vskip 0.1in
\small
\setlength{\tabcolsep}{4.5pt}
\renewcommand{\arraystretch}{1}
\scalebox{1}{
\begin{tabular}{l rrr rrr rrr}
\toprule
\multicolumn{1}{c}{\multirow{2}{*}{\raisebox{-0.6ex}{\textbf{Model}}}}
 & \multicolumn{3}{c}{\textbf{IEMOCAP}}
 & \multicolumn{3}{c}{\textbf{CREMA-D}}
 & \multicolumn{3}{c}{\textbf{ESD}} \\
\cmidrule(lr){2-4} \cmidrule(lr){5-7} \cmidrule(lr){8-10}
 & MEDR &  WEDR &  MSE
 & MEDR &  WEDR &  MSE
 & MEDR &  WEDR &  MSE \\
\midrule
\multicolumn{10}{c}{\textit{\textbf{bitrate = 1.5\,kbps}}} \\
\midrule
EnCodec          & 21.19 & 20.86 &  7.18 & 40.10 & 39.14 & 10.63 & 45.63 & 45.63 & 7.56 \\
DAC              & 17.05 & 17.05 &  9.71 & 40.48 & 39.59 & 10.01 & 46.67 & 46.67 & 6.44 \\
SpeechTokenizer  & 16.24 & 15.92 &  4.74 & 29.42 & 28.76 &  9.72 & 25.55 & 25.55 & 2.86 \\
X-Codec          &  \underline{9.09} &  \underline{9.11} &  \underline{3.80} & \underline{26.72} & \underline{25.84} &  \underline{5.96} &  \underline{21.14} &  \underline{21.14} & \underline{2.38} \\
\textbf{AffectCodec}     &  \textbf{5.27} & \textbf{5.63}     &   \textbf{2.48}    &  \textbf{12.67}   &  \textbf{12.57}   &  \textbf{3.77}   &   \textbf{20.04}    &   \textbf{20.04}    &   \textbf{1.75}  \\
\midrule
\multicolumn{10}{c}{\textit{\textbf{bitrate = 3.0\,kbps}}} \\
\midrule
EnCodec          & 10.16 & 9.82 &  3.97 & 27.25 & 26.56 & 5.43 & 29.25 & 29.25 & 3.29 \\
DAC              & 6.23  & 6.27  &  3.47  & 17.52  & 16.99  & 4.51  & 19.37  & 19.37  &  1.71 \\
SpeechTokenizer  & 7.05  & 7.01  &  2.89  & 15.71  & 15.41  &  4.44  & \textbf{8.75}  & \textbf{8.75}  & \underline{1.16}  \\
X-Codec          & \underline{3.08}  & \underline{2.71}   & \underline{2.02}   & \textbf{8.51}  & \textbf{8.47}  & \textbf{2.57}   & 13.75   & 13.75   & 1.30  \\
\textbf{AffectCodec}      &    \textbf{1.77}   &   \textbf{2.40}    &   \textbf{1.62}    &   \underline{9.25}    &  \underline{8.99}     &  \underline{2.64}     &  \underline{9.66}     &  \underline{9.66}     &  \textbf{0.85}    \\
\midrule
\multicolumn{10}{c}{\textit{\textbf{bitrate = 6.0\,kbps}}} \\
\midrule
EnCodec          & 6.17  &  6.03 &  2.26  & 14.40  & 14.03  & 3.22  & 16.26  & 16.26  &  1.89 \\
DAC              & \underline{1.95}  & \underline{2.06}  &  \underline{0.86}  & \underline{3.19}  & \underline{3.16} & \underline{1.38}  & \underline{6.84}  & \underline{6.84}  &  \underline{0.51} \\
SpeechTokenizer  & -- & -- &  -- & -- & -- &  -- & -- & -- & -- \\
X-Codec          &  -- &  -- &  -- & -- & -- &  -- &  -- &  -- & -- \\
\textbf{AffectCodec}      &  \textbf{0.85}     &  \textbf{0.12}     &   \textbf{0.76}    &   \textbf{1.19}    &   \textbf{1.12}    &  \textbf{1.16}     &   \textbf{2.20}    &  \textbf{2.20}     &   \textbf{0.30}   \\
\bottomrule
\end{tabular}
}
\end{table}


\begin{table}[t]
\centering
\caption{Acoustic quality and intelligibility across bitrates and datasets. 
ViSQOL estimates perceptual quality on a MOS-like scale (higher is better). 
STOI measures short-time objective intelligibility (higher is better). 
WER (\%) reports word error rate from Whisper-Large-v3~\cite{whisper} (lower is better).}

\label{tab:acoustic}
\vskip 0.1in
\small
\setlength{\tabcolsep}{4.5pt}
\renewcommand{\arraystretch}{1}
\scalebox{1}{
\begin{tabular}{l rrr rrr rrr}
\toprule
\multicolumn{1}{c}{\multirow{2}{*}{\raisebox{-0.6ex}{\textbf{Model}}}}
 & \multicolumn{3}{c}{\textbf{IEMOCAP}}
 & \multicolumn{3}{c}{\textbf{CREMA-D}}
 & \multicolumn{3}{c}{\textbf{ESD}} \\
\cmidrule(lr){2-4} \cmidrule(lr){5-7} \cmidrule(lr){8-10}
 & ViSQOL & STOI & WER
 & ViSQOL & STOI & WER
 & ViSQOL & STOI & WER \\
\midrule
\multicolumn{10}{c}{\textit{\textbf{bitrate = 1.5\,kbps}}} \\
\midrule
EnCodec          & 2.61 & 0.659 & 37.69 & 2.63 & 0.652 & 31.88 & 3.18 & 0.785 & 14.63 \\
DAC              & 2.42 & 0.653 & 29.68 & 2.80 & \underline{0.661} & 21.49 & 3.20 & 0.761 & 10.48 \\
SpeechTokenizer  & 2.64 & 0.623 & 23.47 & 2.41 & 0.593 & 24.28 & 3.45 & 0.816 &  7.54 \\
X-Codec          & \underline{3.07} & \underline{0.696} &  \textbf{9.54} & \underline{2.95} & 0.656 &  \textbf{5.09} & \textbf{3.92} & \underline{0.859} &  \textbf{3.46} \\
\textbf{AffectCodec} & \textbf{3.31} & \textbf{0.730} & \underline{15.39} & \textbf{3.14} & \textbf{0.707} & \underline{12.62} & \underline{3.82} & \textbf{0.872} &  \underline{6.04} \\
\midrule
\multicolumn{10}{c}{\textit{\textbf{bitrate = 3.0\,kbps}}} \\
\midrule
EnCodec          & 3.11 & 0.724 & 17.30 & 3.13 & 0.725 & 10.68 & 3.69 & 0.854 &  5.49 \\
DAC              & \underline{3.46} & \underline{0.775} & 10.45 & \underline{3.61} & \underline{0.782} &  4.56 & 3.98 & 0.868 &  3.88 \\
SpeechTokenizer  & 3.30 & 0.700 & 13.18 & 3.13 & 0.670 & 10.06 & 4.02 & \underline{0.886} &  4.03 \\
X-Codec          & 3.35 & 0.728 &  \underline{7.03} & 3.26 & 0.687 &  \textbf{3.29} & \underline{4.13} & 0.875 &  \textbf{2.64} \\
\textbf{AffectCodec} & \textbf{3.82} & \textbf{0.833} &  \textbf{6.99} & \textbf{3.63} & \textbf{0.816} &  \underline{3.44} & \textbf{4.20} & \textbf{0.932} &  \underline{2.77} \\
\midrule
\multicolumn{10}{c}{\textit{\textbf{bitrate = 6.0\,kbps}}} \\
\midrule
EnCodec          & 3.49 & 0.785 &  9.08 & 3.54 & 0.798 &  3.51 & 4.01 & 0.907 &  3.38 \\
DAC              & \textbf{4.37} & \textbf{0.918} &  \textbf{4.42} & \textbf{4.32} & \textbf{0.916} &  \textbf{1.13} & \textbf{4.61} & \textbf{0.969} &  \textbf{1.18} \\
SpeechTokenizer  &   -- &    -- &    -- &   -- &    -- &    -- &   -- &    -- &    -- \\
X-Codec          &   -- &    -- &    -- &   -- &    -- &    -- &   -- &    -- &    -- \\
\textbf{AffectCodec} & \underline{4.18} & \underline{0.902} &  \underline{4.63} & \underline{4.06} & \underline{0.899} &  \underline{1.83} & \underline{4.44} & \underline{0.963} &  \underline{1.79} \\
\bottomrule
\end{tabular}
}
\end{table}

\paragraph{Acoustic reconstruction quality.}
Table~\ref{tab:acoustic} shows that AffectCodec preserves emotion with only a
limited acoustic trade-off. At 1.5\,kbps, AffectCodec achieves the best
ViSQOL and STOI on IEMOCAP and CREMA-D, and the best STOI on ESD, indicating
that explicit emotion allocation does not undermine low-bitrate acoustic
quality. At 3.0\,kbps, it further obtains the best ViSQOL and STOI across all
three benchmarks, while matching or closely approaching the best WER. Since
X-Codec is explicitly designed for semantic tokenization, its strong WER is
expected; nevertheless, AffectCodec remains second-best or closely competitive
on WER while substantially improving emotion preservation. At 6.0\,kbps, DAC
remains strongest on acoustic metrics, but AffectCodec is consistently
second-best with small gaps in ViSQOL and STOI. These results show that AffectCodec
maintains competitive perceptual quality and intelligibility despite reserving
quantization capacity for affective information, yielding a favorable trade-off
between emotion preservation and acoustic reconstruction.

\begin{table}[t]
\centering
\caption{Ablation study on IEMOCAP at 1.5\,kbps. Each row modifies one
component from the full model. For MEDR, WEDR, and VAD MSE, lower is better; for ViSQOL and STOI, higher is better.}
\label{tab:ablation}
\vskip 0.1in
\small
\setlength{\tabcolsep}{4.5pt}
\renewcommand{\arraystretch}{1.2}
\begin{tabular}{ccc c ccccc}
\toprule
\multicolumn{3}{c}{\textbf{Architecture}}
& 
& \multicolumn{1}{c}{\multirow{2}{*}{\raisebox{-0.6ex}{{MEDR}}}}
& \multicolumn{1}{c}{\multirow{2}{*}{\raisebox{-0.6ex}{{WEDR}}}}
& \multicolumn{1}{c}{\multirow{2}{*}{\raisebox{-0.6ex}{{MSE}}}}
& \multicolumn{1}{c}{\multirow{2}{*}{\raisebox{-0.6ex}{{ViSQOL}}}}
& \multicolumn{1}{c}{\multirow{2}{*}{\raisebox{-0.6ex}{{STOI}}}} \\
\cmidrule(lr){1-3}
\textbf{Quantizer} & \textbf{MRT} & \textbf{CEM}
& & & & & & \\
\midrule
RVQ              & \checkmark & \checkmark & & 14.44 & 14.58 & 6.39 & 2.70 & 0.675 \\
Factorized RFSQ  & \checkmark & \checkmark & & 10.23 & 10.26 & 6.51 & 2.95 & 0.684 \\
BD-RFSQ          &    & \checkmark & & 8.37 & 8.57 & 5.79 & 2.98 & 0.695 \\
BD-RFSQ          & \checkmark &    & & 6.94 & 6.95 & 4.29 & 3.11 & 0.721 \\
\midrule
BD-RFSQ & \checkmark & \checkmark & &
\textbf{5.27} & \textbf{5.63} & \textbf{2.48} & \textbf{3.31} & \textbf{0.730} \\
\bottomrule
\end{tabular}
\end{table}

\subsection{Ablation Study}\label{sec:ablation}

We ablate the main components of AffectCodec on IEMOCAP at 1.5\,kbps, the
most challenging low-bitrate setting. Each variant removes or replaces one
component while keeping the remaining training setup unchanged. As shown in
Table~\ref{tab:ablation}, the full model achieves the best results on both
emotion-preservation metrics and acoustic metrics.

The quantizer design has the largest impact. Replacing BD-RFSQ with standard
RVQ increases MEDR from 5.27\% to 14.44\% and WEDR from 5.63\% to 14.58\%,
showing that conventional residual vector quantization does not provide
sufficient structural protection for emotion-relevant information. Factorized
RFSQ improves over RVQ, reducing MEDR to 10.23\%, but remains clearly worse
than BD-RFSQ. This gap confirms that without the block-diagonal constraint, dominant
reconstruction gradients are free to repurpose emotion dimensions, directly
confirming the cross-stream leakage mechanism identified in
Sec.~\ref{sec:intro}.

MRT and the CEM module provide complementary gains. Removing MRT raises MEDR from
5.27\% to 8.37\% and increases V/A/D MSE from 2.48 to 5.79, indicating that
multi-rate training is important for robust emotion preservation under limited
capacity. Removing the CEM module also degrades performance, increasing MEDR to 6.94\%
and VAD MSE to 4.29, confirming that coarse-level emotion conditioning
supplies useful affective content to the reserved emotion subspace. The larger
drop caused by removing MRT further shows that explicit low-rate supervision is
particularly important at the 1.5\,kbps operating point.

Notably, the full model also achieves the best ViSQOL and STOI among all
ablations. This indicates that the proposed components do not merely trade
acoustic quality for emotion preservation; instead, structurally separating
emotion and acoustic information, enriching the emotion pathway, and training
across rates jointly improve the overall codec representation.

\section{Conclusion}

We have presented AffectCodec, an emotion-preserving neural speech codec that addresses the systematic loss of affective information in existing quantization pipelines. Our core contribution, Block-Diagonal Residual FSQ (BD-RFSQ), structurally isolates emotion and acoustic subspaces within the quantizer, transforming bit allocation from implicit and loss-driven to explicit and architecturally guaranteed. Combined with multi-granularity emotion conditioning and multi-rate training with biased stage dropout, AffectCodec achieves substantial reductions in Emotion Degradation Rate across three benchmarks and all tested bitrates---with the largest gains at $\leq 3$\,kbps where prior codecs degrade most sharply---while maintaining competitive acoustic quality and intelligibility. Thorough ablations confirm that each proposed component contributes meaningfully and that emotion preservation does not come at the expense of reconstruction fidelity. Beyond emotion, BD-RFSQ provides a general and principled mechanism for protecting designated speech attributes under low-bitrate neural compression while preserving the flat-token interface required by speech language models.

\textbf{Limitations} The emotion2vec teacher carries its own biases and may under-represent certain emotion categories; the BD-RFSQ partition sizes and multi-rate stage targets are manually chosen, and automatic attribute--rate allocation remains future work. Our evaluation focuses on 16\,kHz speech and emotion preservation metrics derived from external models; future work may further examine how the preserved emotion information benefits downstream speech language models.

\textbf{Broader Impact.} Improving codec-level emotion fidelity can benefit emotionally aware speech technologies. At the same time, better preservation and generation of affective cues could be misused for emotional manipulation or deceptive synthetic speech. We therefore encourage deployment together with transparency mechanisms, watermarking, and safeguards for responsible use.

{\small
\bibliographystyle{plainnat}
\bibliography{ref}
}

\newpage
\appendix

\section{BD-RFSQ Algorithm}\label{app:algorithm}
 
Algorithm~\ref{alg:bdrfsq} provides pseudocode for the BD-RFSQ forward pass. At inference time, the number of active stages can be truncated to $K' < K$ for lower-bitrate operation without retraining.
 
\begin{algorithm}[h]
\caption{BD-RFSQ Forward Pass}\label{alg:bdrfsq}
\begin{algorithmic}[1]
\Require Emotion features $\mathbf{E} \in \mathbb{R}^{d_e' \times T}$, fused acoustic features $\mathbf{A}_f \in \mathbb{R}^{d_a' \times T}$, number of stages $K$
\Ensure Quantized latent $\widehat{\mathbf{U}} \in \mathbb{R}^{d \times T}$, token sequence $\{I_k\}_{k=1}^{K}$

\State \textcolor{gray}{\textit{\% Latent partitioning}}
\State $\mathbf{U}_e \gets \phi_e(\mathbf{E}) \in \mathbb{R}^{d_e \times T}$ \Comment{Linear: $d_e' \to d_e$}
\State $\mathbf{U}_a \gets \phi_a(\mathbf{A}_f) \in \mathbb{R}^{d_a \times T}$ \Comment{Linear: $d_a' \to d_a$}
\State $\mathbf{U} \gets \mathrm{Concat}(\mathbf{U}_e,\, \mathbf{U}_a)$ \Comment{$d = d_e + d_a$}

\State \textcolor{gray}{\textit{\% Residual iteration}}
\State $\mathbf{r}_0 \gets \mathbf{U}$, \quad $\widehat{\mathbf{U}} \gets \mathbf{0}$

\For{$k = 1, 2, \dots, K$}
    \State \textcolor{gray}{\textit{\% (i) Block-diagonal input projection}}
    \State $\mathbf{z}_k \gets
    \begin{bmatrix} \pi_{\mathrm{in},e}^{(k)} & \mathbf{0} \\ \mathbf{0} & \pi_{\mathrm{in},a}^{(k)} \end{bmatrix}
    \mathbf{r}_{k-1}$ \Comment{$\mathbb{R}^d \to \mathbb{R}^f$}

    \State \textcolor{gray}{\textit{\% (ii) Affine normalization \& scalar quantization}}
    \State $\mathbf{s}_k \gets \mathrm{softplus}(\boldsymbol{\ell}_k) + \epsilon$
    \State $\widetilde{\mathbf{z}}_k \gets \mathbf{s}_k \odot (\mathbf{z}_k - \mathbf{b}_k)$
    \State $\widehat{\widetilde{\mathbf{z}}}_k,\, I_k \gets \mathrm{FSQ}(\widetilde{\mathbf{z}}_k)$

    \State \textcolor{gray}{\textit{\% (iii) Back-projection \& residual update}}
    \State $\widehat{\mathbf{z}}_k \gets \widehat{\widetilde{\mathbf{z}}}_k \oslash \mathbf{s}_k + \mathbf{b}_k$
    \State $\widehat{\mathbf{u}}_k \gets
    \begin{bmatrix} \pi_{\mathrm{out},e}^{(k)} & \mathbf{0} \\ \mathbf{0} & \pi_{\mathrm{out},a}^{(k)} \end{bmatrix}
    \widehat{\mathbf{z}}_k$ \Comment{$\mathbb{R}^f \to \mathbb{R}^d$}
    \State $\mathbf{r}_k \gets \mathbf{r}_{k-1} - \widehat{\mathbf{u}}_k$
    \State $\widehat{\mathbf{U}} \gets \widehat{\mathbf{U}} + \widehat{\mathbf{u}}_k$
\EndFor

\Return $\widehat{\mathbf{U}},\; \{I_k\}_{k=1}^{K}$
\end{algorithmic}
\end{algorithm}

\section{Proof of Structural Partition Guarantee}\label{app:proofs}
 
\begin{proposition}[Block separation invariant]\label{prop:partition}
For any BD-RFSQ forward pass with $K$ stages, the emotion residual
$\mathbf{r}_k^{(1:d_e)}$ depends only on $\mathbf{U}^{(1:d_e)}$ and
the acoustic residual $\mathbf{r}_k^{(d_e{+}1:d)}$ depends only on
$\mathbf{U}^{(d_e{+}1:d)}$, for all $k=0,1,\dots,K$.
\end{proposition}
 
\begin{proof}
We proceed by induction on the stage index $k$.
 
\textbf{Base case} ($k{=}0$). $\mathbf{r}_0 = \mathbf{U} = \mathrm{Concat}(\mathbf{U}_e, \mathbf{U}_a)$, which is trivially block-separated by construction.
 
\textbf{Inductive step.} Assume $\mathbf{r}_{k-1}$ is block-separated, i.e., $\mathbf{r}_{k-1}^{(1:d_e)}$ depends only on $\mathbf{U}^{(1:d_e)}$ and $\mathbf{r}_{k-1}^{(d_e{+}1:d)}$ depends only on $\mathbf{U}^{(d_e{+}1:d)}$. We trace each operation in stage $k$:
 
\begin{enumerate}
    \item \textit{Block-diagonal input projection.}
    $\mathbf{z}_k = \pi_{\mathrm{in}}^{(k)}(\mathbf{r}_{k-1})$, where $\pi_{\mathrm{in}}^{(k)} = \mathrm{diag}(\pi_{\mathrm{in},e}^{(k)}, \pi_{\mathrm{in},a}^{(k)})$. By block-diagonality, $\mathbf{z}_k^{(1:f_e)}$ depends only on $\mathbf{r}_{k-1}^{(1:d_e)}$ and $\mathbf{z}_k^{(f_e{+}1:f)}$ depends only on $\mathbf{r}_{k-1}^{(d_e{+}1:d)}$. Block separation is preserved.
 
    \item \textit{Affine normalization.}
    $\widetilde{\mathbf{z}}_k = \mathbf{s}_k \odot (\mathbf{z}_k - \mathbf{b}_k)$. Both $\odot$ (element-wise multiplication) and subtraction act per-dimension, so no cross-partition mixing occurs.
 
    \item \textit{FSQ quantization.}
    $\widehat{\widetilde{\mathbf{z}}}_k = \mathrm{FSQ}(\widetilde{\mathbf{z}}_k)$. FSQ independently rounds each scalar dimension to its nearest grid point, preserving block separation.
 
    \item \textit{Inverse affine.}
    $\widehat{\mathbf{z}}_k = \widehat{\widetilde{\mathbf{z}}}_k \oslash \mathbf{s}_k + \mathbf{b}_k$. Again per-dimension, preserving separation.
 
    \item \textit{Block-diagonal output projection.}
    $\widehat{\mathbf{u}}_k = \pi_{\mathrm{out}}^{(k)}(\widehat{\mathbf{z}}_k)$, where $\pi_{\mathrm{out}}^{(k)} = \mathrm{diag}(\pi_{\mathrm{out},e}^{(k)}, \pi_{\mathrm{out},a}^{(k)})$. By the same argument as step~1, $\widehat{\mathbf{u}}_k$ is block-separated.
 
    \item \textit{Residual update.}
    $\mathbf{r}_k = \mathbf{r}_{k-1} - \widehat{\mathbf{u}}_k$. Coordinate-wise subtraction of two block-separated vectors yields a block-separated result.
\end{enumerate}
 
By induction, block separation holds at every stage. Since $\widehat{\mathbf{U}} = \sum_{k=1}^{K}\widehat{\mathbf{u}}_k$ is a sum of block-separated vectors, the final quantized output is also block-separated: $\widehat{\mathbf{U}}^{(1:d_e)}$ depends only on $\mathbf{U}^{(1:d_e)}$, and $\widehat{\mathbf{U}}^{(d_e{+}1:d)}$ depends only on $\mathbf{U}^{(d_e{+}1:d)}$.
\end{proof}
 
\textbf{Remark.} This guarantee holds \emph{inside the quantizer}. The acoustic decoder receives the full concatenated $\widehat{\mathbf{U}}$ and may use both partitions jointly for reconstruction, which is by design: the structural separation prevents cross-stream gradient contamination during quantization, while the decoder retains full access for high-fidelity waveform synthesis.

\section{Evidence for Cross-Stream Gradient Leakage}
\label{app:leakage}

Section~\ref{sec:intro} claims that fully connected quantizer projections
allow acoustic reconstruction gradients to colonize emotion-designated
FSQ dimensions. We verify this empirically by training an
\emph{acoustic linear probe} on the emotion partition of two codec
variants that differ only in their projection structure.

\paragraph{Setup.}
We compare a \emph{fully connected} baseline (identical to AffectCodec
but with unconstrained \texttt{WNConv1d(1024$\to$9)} projections, i.e.\
no block-diagonal constraint) against \emph{AffectCodec} (BD-RFSQ,
block-diagonal projections). All other components---encoder, decoder,
CEM module, loss weights, and training data---are identical. For each
model we extract the per-frame emotion partition codes: the first 3
dimensions across all $K{=}8$ residual stages, yielding a
24-dimensional binary feature vector
$\mathbf{x}_t \in \{-1,+1\}^{24}$ per frame. We then fit an OLS
linear regression from $\mathbf{x}_t$ to the corresponding 80-bin
log-mel spectrogram frame $\mathbf{y}_t$, and evaluate $R^2$ on a
held-out test set. A high $R^2$ indicates that the emotion partition
linearly encodes acoustic information, i.e.\ gradient leakage has
occurred. Experiments use 200 LibriSpeech test-clean utterances
(73,007 frames at 50\,Hz). A random baseline (column-wise permutation
of $\mathbf{x}$) serves as a sanity check.

\paragraph{Results.}
Table~\ref{tab:leakage_probe} reports the acoustic probe $R^2$.

\begin{table}[h]
\centering
\caption{Acoustic linear probe $R^2$ on the emotion partition ($d_e{=}3$
dims, $K{=}8$ stages, 24 binary features). Higher $R^2$ indicates more
acoustic information linearly decodable from the emotion partition,
i.e.\ greater gradient leakage.}
\label{tab:leakage_probe}
\small
\begin{tabular}{lccc}
\toprule
\textbf{Model} & \textbf{$R^2$ (global)} & \textbf{$R^2$ (per-bin mean)} & \textbf{$R^2$ (per-bin median)} \\
\midrule
Fully connected (baseline) & 0.0985 & 0.0983 & 0.0992 \\
AffectCodec (BD-RFSQ)      & 0.0196 & 0.0198 & 0.0154 \\
Random                     & $-$0.0005 & $-$0.0005 & $-$0.0005 \\
\bottomrule
\end{tabular}
\end{table}

\paragraph{Analysis.}
The fully connected baseline achieves $R^2{=}0.099$---five times higher
than AffectCodec ($R^2{=}0.020$)---despite using only 24 binary features
and a strictly linear probe. Under these severe constraints, explaining
nearly 10\% of mel spectrogram variance indicates that acoustic
information has been \emph{explicitly and linearly encoded} in the
emotion FSQ dimensions by the reconstruction gradients. Because the two
models differ only in projection structure, the $\Delta R^2{=}0.079$
gap is directly attributable to the absence of a block-diagonal
constraint.

The residual $R^2{\approx}0.02$ in AffectCodec is not leakage but
reflects the \emph{inherent physical correlation} between emotion and
acoustics: affect is expressed through prosody, energy, and spectral
tilt, so even a structurally pure emotion partition will retain some
predictive power over mel features. The random baseline ($R^2{\approx}0$)
confirms that the probe method is well-calibrated and that this residual
is a property of the data rather than a measurement artifact.
Train and test $R^2$ agree to within 0.004 for both models,
ruling out overfitting.

\section{Rate-Distortion Analysis of Emotion Partition Dimensions}
\label{app:rd_search}

The emotion partition configuration $(f_e{=}3,\ L_e{=}2)$ is not an ad
hoc choice but the result of a systematic rate-distortion search. Because
affective features occupy a low-dimensional subspace, the optimal number
of quantization dimensions $d$ and levels per dimension $L$ cannot be
determined by intuition alone.

\paragraph{Search protocol.}
We evaluate all combinations of $d \in \{1,2,3,4\}$ and $L \in \{2,3,4\}$
(12 configurations) with $K{=}2$ residual stages trained in isolation
under identical optimization. Each configuration is scored on held-out data by the mean squared error
(MSE) between the FSQ reconstruction and the original emotion2vec
features, cosine similarity, and per-stage bitrate $R = d \cdot \log_2 L$.
Rather than fixing a single trade-off
parameter $\lambda$, we trace the \emph{Pareto front} and locate the
\emph{knee point} where marginal efficiency (MSE reduction per additional
bit) drops sharply.
The remaining bit budget $(f_a{=}6,\ L_a{=}4)$ is allocated entirely to
the acoustic partition after the emotion configuration is fixed.

\paragraph{Results.}
Table~\ref{tab:pareto_emotion} reports the Pareto-optimal subset;
Fig.~\ref{fig:pareto_rd} shows the front with the selected operating point.

\begin{table}[h]
\centering
\caption{Pareto-optimal FSQ configurations for the emotion partition
($K{=}2$ stages). MSE is computed between the FSQ reconstruction and
the original emotion2vec features on held-out data. Marginal efficiency
($\times10^{-3}$ MSE/bit) is computed relative to the preceding Pareto
point. Knee: $\star$; selected: $\diamond$.}
\label{tab:pareto_emotion}
\small
\begin{tabular}{ccrccr}
\toprule
$d$ & $L$ & \textbf{Bits} & \textbf{MSE} & \textbf{Cos.\ Sim.} &
\textbf{Marg.\ Eff.} \\
\midrule
1 & 2 &  2.0 & 0.4789 & 0.9979 & ---   \\
1 & 3 &  3.2 & 0.4560 & 0.9980 &  19.6 \\
$\star$\,2 & 2 &  4.0 & 0.2059 & 0.9993 & 301.2 \\
$\diamond$\,3 & 2 &  6.0 & 0.1516 & 0.9995 &  27.2 \\
4 & 2 &  8.0 & 0.1040 & 0.9996 &  23.8 \\
3 & 4 & 12.0 & 0.0813 & 0.9997 &   5.7 \\
4 & 3 & 12.7 & 0.0634 & 0.9998 &  26.3 \\
\bottomrule
\end{tabular}
\end{table}

\begin{figure}[h]
    \centering
    \includegraphics[width=0.65\textwidth]{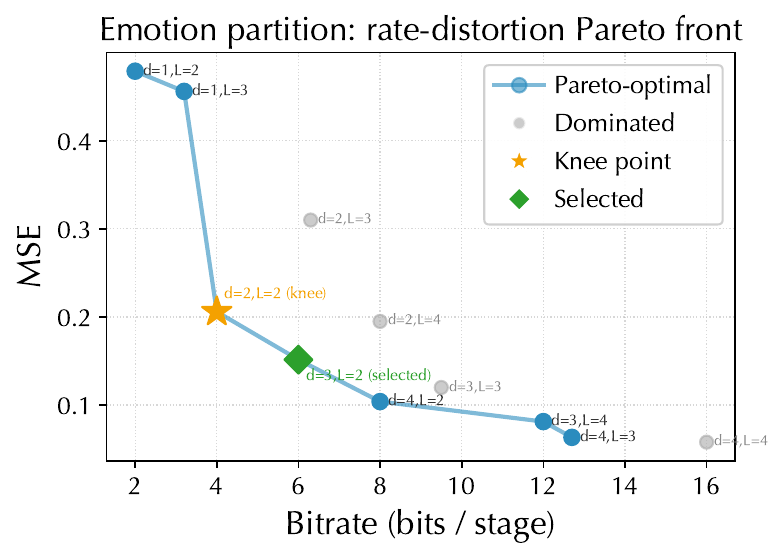}
    \caption{Rate-distortion Pareto front for the emotion FSQ partition.
    Filled markers on the solid line are Pareto-optimal; grey markers
    are dominated configurations. The knee point (yellow star) marks where
    marginal efficiency drops by more than 50\%; the selected operating
    point (green diamond, $d{=}3,\ L{=}2$) lies one step beyond the
    knee to provide a small reconstruction margin at a cost of only
    $+2$\,bits/stage.}
    \label{fig:pareto_rd}
\end{figure}

The knee lies at $(d{=}2,\ L{=}2,\ 4\text{\,bits/stage})$ with marginal
efficiency $301{\times}10^{-3}$\,MSE/bit---an order of magnitude above
all adjacent Pareto points. We select $(d{=}3,\ L{=}2)$, one step
beyond the knee: the additional 2\,bits/stage yields a further 26\%
MSE reduction and raises cosine similarity to 0.9995, providing a
comfortable margin for downstream emotion recognition tasks.
A consistent finding across all configurations is that \emph{increasing
dimensionality $d$ is more effective than increasing levels $L$}: at
comparable bitrates, higher-$d$ configurations reduce MSE by 19--22\%
over higher-$L$ alternatives, suggesting that emotion features occupy a
manifold whose intrinsic dimensionality exceeds the smallest $d$ values
tested. A parallel search on the acoustic branch confirms that acoustic
features are substantially less compressible: cosine similarity exceeds
0.999 for all emotion configurations above 4\,bits, whereas acoustic
features require ${\geq}16$\,bits to reach 0.983---an asymmetry that
justifies the unequal partition $(f_e{=}3, f_a{=}6)$ and motivates
allocating the remaining bit budget entirely to the acoustic partition.

\section{Affine Normalization in BD-RFSQ: Comparison with Prior Work}
\label{app:rfsq_comparison}

We describe how BD-RFSQ addresses residual magnitude decay and contrast
our affine normalization design with the conditioning strategies
proposed in Robust RFSQ~\cite{rrfsq}.

\paragraph{Residual magnitude decay.}
A fundamental obstacle in residual quantization is that successive
residuals shrink in magnitude: once early stages have captured the
dominant signal energy, later stages receive inputs concentrated near
zero and underutilize the fixed FSQ grid. RFSQ~\cite{rrfsq} proposes
two remedies. \emph{Scale conditioning} introduces one learnable global
scalar $\alpha_k$ per stage, rescaling the residual before quantization
and inverting the scaling after. \emph{LayerNorm conditioning} extends
this to per-dimension mean and variance correction, achieving better
distributional regularity; however, to support index-only decoding the
normalization statistics $(\mu_k, \sigma_k)$ must be decoupled from
per-sample inputs, which requires a two-phase procedure: train the
model to convergence, estimate $\mu_k$ and $\sigma_k$ from the
resulting residual distribution, then freeze them for inference.

\paragraph{Our approach: fully learnable per-dimension affine.}
We adopt a per-dimension learnable affine transformation that resolves
both the expressiveness gap of scale conditioning and the two-phase
complication of LayerNorm conditioning:
\begin{equation}
    \widetilde{\mathbf{z}}_k
    = \mathbf{s}_k \odot (\mathbf{z}_k - \mathbf{b}_k),
    \qquad
    \mathbf{s}_k = \mathrm{softplus}(\boldsymbol{\ell}_k) + \epsilon,
    \quad \epsilon = 0.1,
    \label{eq:affine_norm}
\end{equation}
where $\boldsymbol{\ell}_k \in \mathbb{R}^{f}$ and
$\mathbf{b}_k \in \mathbb{R}^{f}$ are jointly optimized throughout
training, initialized to $\mathbf{0}$ (identity at initialization).
The inverse de-normalization $\widehat{\widetilde{\mathbf{z}}}_k
\oslash \mathbf{s}_k + \mathbf{b}_k$ relies only on these fixed model
parameters---not on any runtime input statistics---so the bitstream
carries no side information and index-only decoding (tokens
$\to$ waveform) is supported exactly without any modification to the
codec interface.

\paragraph{Advantages over RFSQ conditioning.}
Table~\ref{tab:rfsq_comparison} summarizes the comparison.
Relative to scale conditioning, our formulation applies independent
correction to each FSQ dimension, correcting both inter-stage magnitude
decay and per-dimension variance imbalance within each stage; scale
conditioning applies only a single global rescaling and cannot correct
distributional shape. Relative to LayerNorm conditioning, our parameters
are end-to-end trainable: they adapt continuously as the encoder and
decoder co-evolve, whereas frozen statistics capture only the residual
distribution at the checkpoint where they were estimated. Both methods
support index-only decoding, but ours does so within a single-phase
training procedure. The $\epsilon{=}0.1$ floor in
Eq.~\eqref{eq:affine_norm} ensures $s_k \geq 0.1$ without a hard
clamp, and $\mathrm{softplus}$ provides non-zero gradient everywhere
($\partial s_k / \partial \boldsymbol{\ell}_k = \sigma(\boldsymbol{\ell}_k) > 0$),
preventing gradient stagnation that would occur at the boundaries of an
$\exp$-clamp parameterization.

\begin{table}[h]
\centering
\caption{Comparison of residual FSQ conditioning strategies for
addressing residual magnitude decay.}
\label{tab:rfsq_comparison}
\small
\begin{tabular}{lccc}
\toprule
\textbf{Property}
  & \textbf{RFSQ-Scale}
  & \textbf{RFSQ-LN}
  & \textbf{BD-RFSQ (ours)} \\
\midrule
Per-dimension scale correction & \texttimes & \checkmark & \checkmark \\
Per-dimension mean correction  & \texttimes & \checkmark & \checkmark \\
End-to-end trainable           & \checkmark & \texttimes~(frozen) & \checkmark \\
Single-phase training          & \checkmark & \texttimes & \checkmark \\
Index-only decoding            & \checkmark & \checkmark$^\dagger$ & \checkmark \\
\bottomrule
\multicolumn{4}{l}{$^\dagger$ Requires statistics to be frozen
post-training; not compatible with continued fine-tuning.}
\end{tabular}
\end{table}

\paragraph{Algorithm.}
Algorithm~\ref{alg:affine} details the affine normalization sub-routine
within each BD-RFSQ stage and the index-only decoding path that it
enables. The full BD-RFSQ forward pass is given in
Algorithm~\ref{alg:bdrfsq}.

\begin{algorithm}[h]
\caption{BD-RFSQ Affine Normalization and Index-Only Decoding}
\label{alg:affine}
\begin{algorithmic}[1]
\Require Projected residual $\mathbf{z}_k \in \mathbb{R}^{f \times T}$,
         learnable $\boldsymbol{\ell}_k \in \mathbb{R}^{f}$,
         $\mathbf{b}_k \in \mathbb{R}^{f}$,
         floor $\epsilon = 0.1$
\Ensure  Quantized reconstruction $\widehat{\mathbf{z}}_k$,
         code index $I_k$

\State \textcolor{gray}{\textit{\% Forward: normalize and quantize}}
\State $\mathbf{s}_k \gets \mathrm{softplus}(\boldsymbol{\ell}_k) + \epsilon$
       \Comment{per-dimension positive scale}
\State $\widetilde{\mathbf{z}}_k \gets \mathbf{s}_k \odot (\mathbf{z}_k - \mathbf{b}_k)$
       \Comment{affine normalization}
\State $\widehat{\widetilde{\mathbf{z}}}_k,\; I_k \gets \mathrm{FSQ}(\widetilde{\mathbf{z}}_k)$
       \Comment{scalar quantization, emits index $I_k$}

\State \textcolor{gray}{\textit{\% Inverse: de-normalize using fixed model parameters only}}
\State $\widehat{\mathbf{z}}_k \gets \widehat{\widetilde{\mathbf{z}}}_k \oslash \mathbf{s}_k + \mathbf{b}_k$
       \Comment{exact inverse; no per-sample statistics required}

\State \textcolor{gray}{\textit{\% Index-only decoding path (tokens $\to$ waveform)}}
\State \textbf{given} index $I_k$:
\State \quad $\widehat{\widetilde{\mathbf{z}}}_k \gets \mathrm{FSQ.indices\_to\_codes}(I_k)$
\State \quad $\widehat{\mathbf{z}}_k \gets \widehat{\widetilde{\mathbf{z}}}_k \oslash \mathbf{s}_k + \mathbf{b}_k$
       \Comment{$\mathbf{s}_k$, $\mathbf{b}_k$ are fixed decoder weights}

\Return $\widehat{\mathbf{z}}_k,\; I_k$
\end{algorithmic}
\end{algorithm}

\section{Additional Implementation Details}\label{app:details}

\subsection{Architecture Details}
 
\textbf{Acoustic backbone.}
We adopt the DAC~\cite{dac} encoder--decoder architecture with Snake activations~\cite{bigvgan}. The encoder uses convolutional blocks with strides $[2, 4, 5, 8]$ (total $320\times$ downsampling), yielding $T = L/320$ frames at 50\,Hz for 16\,kHz input. The base encoder dimension is $d_{\mathrm{enc}} = 64$, expanding to $d_a' = 64 \times 2^4 = 1024$ at the bottleneck. The decoder mirrors the encoder with transposed convolutions and strides $[8, 5, 4, 2]$, with a base dimension of 1536.
 
\textbf{Emotion encoder.}
We use emotion2vec-large~\cite{emotion2vec} as the frozen emotion teacher (hidden dimension $d_e' = 1024$, frame rate 50\,Hz matching the codec). The emotion2vec parameters are frozen throughout training and excluded from the saved checkpoint. A lightweight CNN adapter (following the XCodec~\cite{xcodec} encoder architecture with input and output channels both equal to 1024) aligns the emotion2vec features to the codec frame rate.
 
\textbf{Coarse-granularity Emotion Modulation (CEM) module.}
The attentive pooling layer consists of a two-layer MLP ($1024 \to 1024 \to 1$) with Tanh activation, followed by softmax-weighted aggregation over the time axis. The FiLM projections $g(\cdot)$ and $h(\cdot)$ are single linear layers ($1024 \to 1024$). We initialize $g$ such that $\gamma \approx \mathbf{1}$ and $h$ such that $\beta \approx \mathbf{0}$ (identity initialization).
 
\textbf{BD-RFSQ quantizer.}
$K = 8$ residual stages, input partition $(d_e, d_a) = (256, 768)$, FSQ partition $(f_e, f_a) = (3, 6)$, scalar levels $\mathbf{L} = [2, 2, 2, 4, 4, 4, 4, 4, 4]$. We set $\mathrm{preserve\_symmetry} = \mathrm{True}$ for $L{=}2$ dimensions. Commitment loss weight $\alpha = 0.25$. All block-diagonal projections are implemented as weight-normalized 1$\times$1 convolutions.
 
\textbf{Softplus affine normalization.}
Scale $\mathbf{s}_k = \mathrm{softplus}(\boldsymbol{\ell}_k) + 0.1$, where $\boldsymbol{\ell}_k \in \mathbb{R}^{f}$ is initialized to zero (yielding $\mathbf{s}_k^{(\mathrm{init})} \approx 0.793$). Bias $\mathbf{b}_k \in \mathbb{R}^{f}$ is zero-initialized.
 
\textbf{Emotion decoder.}
A lightweight CNN decoder (following the XCodec~\cite{xcodec} decoder architecture) maps the quantized emotion partition $\widehat{\mathbf{U}}_{1:d_e}$ back to the emotion2vec feature space for the emotion reconstruction loss. A linear layer ($256 \to 1024$) precedes the CNN decoder.
 
\textbf{Post-quantization projection.}
A linear layer ($1024 \to 1024$) is applied to the full quantized latent $\widehat{\mathbf{U}}$ before feeding it to the acoustic decoder, providing additional capacity to mix the separately quantized emotion and acoustic representations for reconstruction.

\subsection{Training Details}

\textbf{Optimizer.}
Generator and discriminator are each trained with AdamW using separate optimizers.
The scheduler follows an exponential learning rate decay schedule.

\textbf{Schedule.}
Total 250K steps, batch size 12.
All losses are active throughout training with fixed weights; no phased warm-up is applied.

\textbf{Discriminator.}
Multi-scale STFT discriminator with feature matching, identical to DAC~\cite{dac}.
The discriminator is trained jointly with the generator using separate optimizers.

\textbf{Hardware.}
4$\times$NVIDIA RTX4090 24\,GB GPUs, approximately 72 hours wall-clock time.

\subsection{Loss Weights}\label{app:loss_weights}

Table~\ref{tab:loss_weights} lists all loss weighting coefficients.
All weights are fixed throughout training; no dynamic balancing is applied.

\begin{table}[h]
\centering
\caption{Loss weighting coefficients.}
\label{tab:loss_weights}
\small
\begin{tabular}{lc}
\toprule
\textbf{Loss term} & \textbf{Weight} \\
\midrule
Mel spectrogram $\mathcal{L}_{\mathrm{mel}}$        & 15.0 \\
GAN generator $\mathcal{L}_{\mathrm{gen}}$           & 1.0  \\
Feature matching $\mathcal{L}_{\mathrm{feat}}$       & 2.0  \\
Commitment $\mathcal{L}_{\mathrm{cm}}$               & 0.25 \\
Emotion feature $\mathcal{L}_{\mathrm{emo}}$         & 25.0 \\
Emotion cycle $\mathcal{L}_{\mathrm{cycle}}$         & 25.0 \\
Multi-rate $\mathcal{L}_{\mathrm{mr}}$               & 1.0  \\
\bottomrule
\end{tabular}
\end{table}

\subsection{Multi-Rate and Biased Dropout Configuration}\label{app:mr_config}
 
Table~\ref{tab:mr_config} specifies the multi-rate supervision targets, per-rate loss weights, and biased stage dropout distribution.
 
\begin{table}[h]
\centering
\caption{Multi-rate training and biased stage dropout configuration.}
\label{tab:mr_config}
\small
\begin{tabular}{ccccc}
\toprule
\textbf{Target stages} $m$ & \textbf{Bitrate} & \textbf{Mel weight} $w_m^{\mathrm{mel}}$ & \textbf{Cycle weight} $w_m^{\mathrm{cycle}}$ & \textbf{Dropout prob.} \\
\midrule
2  & 1.5\,kbps & 0.5 & 0.5 & 0.50 \\
4  & 3.0\,kbps & 0.3 & 0.3 & 0.30 \\
8  & 6.0\,kbps & 0.0 & 0.0 & 0.20 \\
\bottomrule
\end{tabular}
\vspace{0.5em}
 
\raggedright
\small
\textit{Note:} Stage 8 (full model) receives zero multi-rate loss weight because it is already supervised by the main-path reconstruction and cycle losses. The quantizer dropout probability is 0.75: 75\% of training samples draw a stage count from the categorical distribution above, while the remaining 25\% use all $K{=}8$ stages.
\end{table}

\section{Evaluation Protocol Details}\label{app:eval}
 
\subsection{SER Classifiers for EDR}
 
To avoid circular evaluation with the emotion2vec teacher used during training, we compute Emotion Degradation Rate (EDR) using three independently trained Speech Emotion Recognition (SER) classifiers, following the S3PRL/SUPERB benchmark design~\cite{superb}:
 
\begin{enumerate}
    \item \textbf{HuBERT-Large}~\cite{hubert}: frozen features from the last hidden layer, followed by a mean-pooling layer and a two-layer MLP classifier.
    \item \textbf{WavLM-Large}~\cite{wavlm}: same downstream architecture as above.
    \item \textbf{Wav2Vec\,2.0-Large}~\cite{wav2vec2}: same downstream architecture as above.
\end{enumerate}
 
Each classifier is trained on the emotion labels of the respective dataset (IEMOCAP, CREMA-D, or ESD) and applied to both original and codec-reconstructed speech. The reported EDR is the average across the three classifiers.
 
\subsection{V/A/D MSE}
 
Valence--Arousal--Dominance (V/A/D) MSE captures continuous affective distortion that categorical EDR may miss. We use a Wav2Vec\,2.0-Large model fine-tuned on MSP-Podcast~\cite{msp} to predict V/A/D values in $[0, 1]$ for both original and reconstructed speech. The MSE is computed as
\begin{equation}
    \mathrm{MSE}_{\mathrm{VAD}} = \frac{1}{3N}\sum_{i=1}^{N} \sum_{d \in \{V,A,D\}} \left( \hat{y}_d^{(i)} - y_d^{(i)} \right)^2,
\end{equation}
where $y_d^{(i)}$ and $\hat{y}_d^{(i)}$ are the predicted values from original and reconstructed speech respectively. We report MSE $\times 10^{-3}$ in all tables for readability.

\section{Bitrate Calculation}\label{app:bitrate}
 
BD-RFSQ uses FSQ levels $\mathbf{L} = [2,2,2,4,4,4,4,4,4]$ at each of $K$ residual stages. Each frame at each stage produces a single composite index with
\begin{equation}
    C = \prod_{j=1}^{f} L_j = 2^3 \times 4^6 = 8 \times 4096 = 32{,}768 = 2^{15} \text{ possible codes.}
\end{equation}
This requires $\log_2(2^{15}) = 15$ bits per frame per stage. Within these 15 bits, the emotion partition contributes $\log_2(2^3) = 3$ bits and the acoustic partition contributes $\log_2(4^6) = 12$ bits. This 3:12 ratio is \emph{structurally fixed} at every stage by the block-diagonal design.
 
With a frame rate of 50\,Hz, the bitrate at $K'$ active stages is
\begin{equation}
    \text{Bitrate} = 15 \times K' \times 50 = 750 \, K' \;\text{bps}.
\end{equation}
 
Table~\ref{tab:bitrate} lists the operating bitrates used in this work and the corresponding emotion/acoustic bit allocation.
 
\begin{table}[h]
\centering
\caption{Bitrate configuration at different numbers of active stages.}
\label{tab:bitrate}
\small
\begin{tabular}{ccccc}
\toprule
\textbf{Active stages} $K'$ & \textbf{Total (kbps)} & \textbf{Emo bits/frame} & \textbf{Aco bits/frame} & \textbf{Emo ratio} \\
\midrule
2 & 1.5 & 6 & 24 & 20\% \\
4 & 3.0 & 12 & 48 & 20\% \\
8 & 6.0 & 24 & 96 & 20\% \\
\bottomrule
\end{tabular}
\end{table}

\section{Computational Cost}\label{app:compute}

Table~\ref{tab:params} reports the measured parameter count of each module in
AffectCodec and comparisons with baselines. All counts are obtained by running
\texttt{sum(p.numel() for p in model.parameters())} on the respective
official checkpoints.

\begin{table}[h]
\centering
\caption{Parameter count breakdown. All values measured from loaded checkpoints.}
\label{tab:params}
\small
\begin{tabular}{lrrr}
\toprule
\textbf{Module} & \textbf{Total (M)} & \textbf{Trainable (M)} & \textbf{Frozen (M)} \\
\midrule
Acoustic encoder (DAC)                        & 21.52 & 21.52 & 0.00 \\
Acoustic decoder (DAC)                        & 52.33 & 52.33 & 0.00 \\
Emotion CNN adapter -- encoder                & 26.22 & 26.22 & 0.00 \\
Emotion CNN adapter -- decoder                & 29.36 & 29.36 & 0.00 \\
CEM module (AttentivePool + FiLM)             &  3.15 &  3.15 & 0.00 \\
Pre-quantization projections ($\phi_e$, $\phi_a$) & 1.05 & 1.05 & 0.00 \\
BD-RFSQ ($K{=}8$ stages)                     &  0.10 &  0.10 & 0.00 \\
Post-quantization projections                 &  1.31 &  1.31 & 0.00 \\
\midrule
\textbf{Total (trainable)}                    & \textbf{135.05} & \textbf{135.05} & \textbf{0.00} \\
emotion2vec-large (frozen)                    & 164.05 & 0.00 & 164.05 \\
\midrule
\textbf{Total (loaded)}                       & \textbf{299.10} & \textbf{135.05} & \textbf{164.05} \\
\bottomrule
\end{tabular}
\end{table}

\textbf{Comparison with baselines.}
The trainable parameter count of AffectCodec (135.1\,M) is larger than DAC
(74.2\,M) but smaller than X-Codec (160.7\,M), primarily due to the emotion
CNN adapter (encoder + decoder, 55.6\,M combined), which follows the
XCodec~\cite{xcodec} architecture and constitutes the dominant additional cost
over the DAC backbone. The BD-RFSQ quantizer itself adds only 0.1\,M
parameters, as its projections operate in a compact 9-dimensional FSQ space.
The frozen emotion2vec encoder (164.0\,M) adds memory overhead during training
but can be discarded at inference when decoding from token indices, since only
the BD-RFSQ lookup, post-quantization projection, and acoustic decoder are
required.

\textbf{Inference cost.}
At inference, the primary additional cost over a standard DAC-like codec is
(1) one forward pass through the frozen emotion2vec encoder and CNN adapter,
and (2) the CEM module (AttentivePool + FiLM, 3.1\,M). The block-diagonal
projections in BD-RFSQ are $1{\times}1$ convolutions in a 9-dimensional space
and contribute negligible compute. When decoding from token indices (e.g., in
a speech language model pipeline), only the BD-RFSQ index-to-code lookup,
post-quantization projection, and waveform decoder are executed.

\section{Baseline Configuration Details}\label{app:baselines}
 
All baselines use official pretrained checkpoints and are evaluated without fine-tuning. Table~\ref{tab:baselines} summarizes the key configuration of each baseline.
 
\begin{table}[h]
\centering
\caption{Baseline codec configurations.}
\label{tab:baselines}
\small
\setlength{\tabcolsep}{3.5pt}
\begin{tabular}{lcccccc}
\toprule
\textbf{Model} & \textbf{Quantizer} & \textbf{Stages} & \textbf{Codebook} & \textbf{Frame rate} & \textbf{Max kbps} & \textbf{Emo.\ mech.} \\
\midrule
EnCodec       & RVQ    & 8  & 1024 & 75\,Hz & 6.0  & None \\
DAC           & RVQ    & 9  & 1024 & 86\,Hz & 8.0  & None \\
SpeechToken.  & RVQ    & 8  & 1024 & 50\,Hz & 4.0  & None \\
X-Codec       & RVQ    & 8  & 1024 & 50\,Hz & 4.0  & None \\
\midrule
AffectCodec   & BD-RFSQ & 8 & $2^{15}$ & 50\,Hz & 6.0  & BD-RFSQ+CEM \\
\bottomrule
\end{tabular}
\end{table}
 
For fair bitrate comparison, we truncate the number of active RVQ/RFSQ stages at inference. EnCodec and DAC support up to 6.0\,kbps with comparable stage truncation; SpeechTokenizer and X-Codec both support up to 4.0\,kbps. SpeechTokenizer and X-Codec entries at 6.0\,kbps are marked ``--'' in the main tables as they do not support this bitrate.

\section{Reproducibility}\label{app:reproduce}
 
We summarize the key information for reproducing the main results:
 
\begin{itemize}
    \item \textbf{Code:} We will release the full training code upon acceptance, including the BD-RFSQ quantizer, CEM module, and multi-rate training loop.
 
    \item \textbf{Training data:} LibriSpeech~\cite{librispeech} is publicly available. IEMOCAP~\cite{iemocap} requires a license agreement from USC.
 
    \item \textbf{Evaluation data:} All three evaluation benchmarks (IEMOCAP, CREMA-D~\cite{crema}, ESD~\cite{esd}) are publicly available. CREMA-D and ESD can be downloaded without restrictions; IEMOCAP requires a license.
 
    \item \textbf{Pretrained dependencies:} emotion2vec-large is available from ModelScope/FunASR. HuBERT-Large, WavLM-Large, and Wav2Vec\,2.0-Large are available from HuggingFace. Whisper-Large-v3 is available from OpenAI.
 
    \item \textbf{Compute:} 4$\times$RTX 4090 GPUs, $\sim$72 hours. A single-GPU configuration with proportionally smaller batch size and longer training is expected to produce comparable results.
\end{itemize}

\end{document}